\documentclass[runningheads]{llncs}

\usepackage[T1]{fontenc}

\usepackage{graphicx}

\usepackage{algorithm} 
\usepackage{algorithmic} 
\usepackage{microtype}
\usepackage{graphicx}
\usepackage{caption}
\usepackage{subcaption}
\usepackage{booktabs} 

\usepackage{amsmath}
\usepackage{amssymb}
\usepackage{mathtools}
\usepackage{amsthm}
\usepackage{amsfonts}
\theoremstyle{plain}
\usepackage{enumitem}
\usepackage{algorithm}
\usepackage{algorithmic}

\theoremstyle{definition}

\theoremstyle{remark}
\usepackage[textsize=tiny]{todonotes}
\usepackage{bm} 

\usepackage{auxlib}

\makeatletter
\newcommand{\printfnsymbol}[1]{%
  \textsuperscript{\@fnsymbol{#1}}%
}
\makeatother

\title{Sponsored Search Auction Design Beyond Single Utility Maximization}

\author{Changfeng Xu \and  Chao Peng\thanks{Corresponding authors.} \and Chenyang Xu$^*$ \and Zhengfeng Yang}

\authorrunning{Changfeng Xu, Chao Peng, Chenyang Xu, and Zhengfeng Yang}
\institute{Shanghai Key Laboratory of Trustworthy Computing, Software Engineering Institute, East China Normal University, China\\
\email{51265902120@stu.ecnu.edu.cn,
\{cpeng, cyxu, zfyang\}@sei.ecnu.edu.cn}
}

\begin{document}

\maketitle              

\begin{abstract}

Auction design for the modern advertising market has gained significant prominence in the field of game theory. With the recent rise of auto-bidding tools, an increasing number of advertisers in the market are utilizing these tools for auctions. The diverse array of auto-bidding tools has made auction design more challenging. Various types of bidders, such as quasi-linear utility maximizers and constrained value maximizers, coexist within this dynamic gaming environment. We study sponsored search auction design in such a mixed-bidder world and aim to design truthful mechanisms that maximize the total social welfare. To simultaneously capture the classical utility and the value-max utility, we introduce an allowance utility model. In this model, each bidder is endowed with an additional allowance parameter, signifying the threshold up to which the bidder can maintain a value-max strategy. The paper distinguishes two settings based on the accessibility of the allowance information. In the case where each bidder's allowance is public, we demonstrate the existence of a truthful mechanism achieving an approximation ratio of $(1+\epsilon)$ for any $\epsilon > 0$. In the more challenging private allowance setting, we establish that a truthful mechanism can achieve a constant approximation. Further, we consider uniform-price auction design in large markets and give a truthful mechanism that sets a uniform price in a random manner and admits bounded approximation in expectation.






\end{abstract}

\section{Introduction}

Sponsored search advertising is one of the most classical and fundamental problems in game theory~\cite{DBLP:conf/www/DengLX20,DBLP:conf/wine/LiMM10}. In this context, advertisers bid for the placement of their ads on search engine result pages, and the auction mechanism determines the ad ranking and the cost per click for each advertiser. The design of sponsored search auctions holds considerable significance in online advertising and search engine monetization~\cite{DBLP:conf/ijcai/BirmpasCCL22}. The efficiency and fairness of the auction mechanisms directly impact the advertising ecosystem.

This paper studies sponsored search auction design in the modern advertising market. Compared to the traditional auction environment, where bidders universally acknowledge quasi-linear utilities,
the utility functions of bidders in the modern market are more complex and diverse. As highlighted in recent economical literature~\cite{baisa2017auction}, an increasing number of advertisers are making use of auto-bidding tools~\cite{DBLP:conf/www/WilkensCN17,DBLP:conf/nips/BalseiroDMMZ21} to optimize their profits, posing challenges in devising a model that fits reality well.

Auction design for value maximizers has emerged as a pivotal model to navigate this dynamic landscape. Recent literature~\cite{DBLP:journals/corr/WilkensCN16,DBLP:conf/icml/0002W0S0YD23} has extensively explored this domain. Bidders in the value-max model strive to maximize their total obtained values under various constraints, such as budget constraints~\cite{DBLP:conf/sigecom/BalseiroDMMZ22,DBLP:journals/mor/CaragiannisV21,DBLP:conf/sigecom/DevanurW17} and return-on-investment (ROI) constraints~\cite{szymanski2006impact,DBLP:conf/adkdd/AuerbachGS08,DBLP:conf/wine/LuXZ23}. Experimental studies ~\cite{DBLP:journals/corr/WilkensCN16,DBLP:conf/nips/DengZ21} further highlight the effectiveness of the value maximizer model in accurately capturing the behaviors of a significant number of auto-bidding advertisers.

However, as argued in~\cite{DBLP:conf/aaai/LvZZLYLCW23}, the real-world market is a heterogeneous environment where various utility functions are at play. The classical quasi-linear utility model and the emerging value-max model each capture only a subset of advertisers. Recognizing this diversity, the authors initiate the study of a mixed-bidder environment and propose truthful mechanisms designed to accommodate both quasi-linear utility maximizers and value maximizers. We continue along this line of research and make a further generalization of their setting.

\subsection{Our Contributions}

This paper explores sponsored search auction design for mixed bidders. To gracefully model both quasi-linear bidders and (constrained) value maximizers, a novel \emph{allowance utility function} is introduced. In this new model, each bidder is endowed with an additional allowance. If a bidder's payment is less than the allowance, she operates as a value maximizer -- meaning her utility equals the obtained value. Conversely, if a bidder's payment exceeds the allowance, her utility becomes the obtained value minus the portion of payment exceeding the allowance. See a formal definition in~\cref{eq:utility}.

We study sponsored search auction design for allowance utility maximizers and design truthful mechanisms that maximize the total social welfare. The utility function captures and generalizes the two prominent utility types in the literature. A quasi-linear bidder corresponds to the scenario where the allowance is set to $0$, while a value-max bidder can be considered when the allowance is set to $\infty$. Besides, our model allows for more diverse choices of allowance values, accommodating a broader range of advertiser behaviors in the market. 

Refer to the scenario where the allowance information is publicly disclosed as the \emph{public allowance} setting and the setting where it is not disclosed as the \emph{private allowance} setting. Both public allowance and private settings are considered in the paper. The main results of the paper are summarized as follows:
\begin{itemize}
    \item  For the public allowance setting, we first establish the two crucial properties essential for truthful mechanisms -- \emph{allocation monotonicity} and \emph{unit-price strict-monotonicity} (\cref{lem:public_truthful}). Leveraging them can prove that no deterministic truthful mechanism can guarantee $(1+\epsilon)$ approximation for arbitrarily tunable $\epsilon$.  We then give a deterministic truthful mechanism parameterized by $\epsilon$ that admits an approximation ratio of $(1+\epsilon)$ (\cref{thm:public}). 
    

    \item For the private allowance setting, the multi-parameter nature poses significant challenges to the design of truthful mechanisms. Thus, we explore a class of allowance-independent mechanisms and show that even if not using the reported allowance profile, there exists a randomized mechanism that is truthful and admits a constant approximation (\cref{thm:private_constant}).

    \item We further consider uniform price auction design in large markets, where a single bidder's contribution (power) to the total market is very small. We show that there exists a randomized truthful mechanism (for the private allowance setting) capable of setting a uniform price to obtain an expected social welfare $\frac{3}{8} \left( 1-\sqrt{1-\frac{\rho}{3k}} \right)^2 \cdot \opt$,  where $k$ is the number of slots, $\opt$ is the optimal social welfare and $\rho$ is the ratio between $\opt$ and the maximum contribution of a single bidder (\cref{thm:large_market_poly}). When $\rho=O(k^{2/3})$, the mechanism is $O(k^{2/3})$-approximation.
    

\end{itemize}

\subsection{Paper Organization}

\cref{sec:pre} provides a formal definition of our model and introduces the terminology necessary to understand the paper. Then we present our results for the public allowance setting in~\cref{sec:pub}, the private allowance setting in~\cref{sec:private} and the uniform-price auction in~\cref{sec:one_price}.~\cref{sec:con} concludes the paper and outlines potential avenues for future work.

\section{Preliminaries}\label{sec:pre}

We consider the standard sponsored search environment, where the goods for sale are the $k$ ``slots'' for sponsored links on a search results page. The bidder set in this setting consists of $n$ advertisers, each of whom has a standing bid on the keyword relevant to the search. Each slot $j\in [k]$ is associated with click-through-rate (CTR) $\alpha_j \geq 0$, representing the probability that the end user clicks on this slot. Without loss of generality, assume that $\alpha_1\geq \alpha_2 \geq ...\geq \alpha_k$ and $n\geq k$. 

\paragraph{Allowance Utility Maximizer.} Each bidder $i$ desires to purchase at most one slot and has a private value $v_i>0$ for one click, implying that the maximum amount she is willing to spend for slot $j$ is $v_i\alpha_j$. Additionally, our model introduces an allowance $\gamma_i$ for each bidder $i$, and defines the allowance utility function as follows: if the auctioneer assigns slot $j$ to bidder $i$ and charges her $p_i$, the utility
\begin{equation}\label{eq:utility}
    u_i:= \left\{
\begin{aligned}
v_i \alpha_j & -(p_i-\gamma_i)^+ \;\;\;\; & \text{if $p_i\leq v_i\alpha_j$ ,}   \\
&-\infty\;\; & \text{otherwise,} \\
\end{aligned}
\right.
\end{equation}
where $(p_i-\gamma_i)^+ := \max\{0, p_i-\gamma_i\}$.

During the auction, each bidder $i$ confidentially communicates her private information to the auctioneer (with the possibility of misreporting). Subsequently, the auctioneer determines the allocation of slots to bidders $\Pi=\{\pi_1,...,\pi_n\}$ and the corresponding payment requirements $\cP=\{p_1,...,p_n\}$, where $\pi_i,p_i$ represent the assigned slot and the charged money for bidder $i$, respectively. To facilitate the representation of bidders $i$ who are not allocated any slot, we add a dummy slot $k+1$ with CTR $0$ and let $\pi_i=k+1$. Note that each slot except the dummy slot can be assigned to at most one bidder.
The goal of the auctioneer is to design a truthful mechanism that maximizes the total social welfare $\sum_{i\in [n]} v_i \alpha_{\pi_i}$.

\begin{definition}(\cite{DBLP:journals/eatcs/RoughgardenI17})\label{def:DSIC}
	A mechanism is \emph{truthful} (or DSIC) if for any bidder $i$, reporting the private information truthfully always maximizes the utility regardless of other bidders' profiles, and the utility of any truthtelling bidder is non-negative. 
\end{definition}

\section{Public Allowance}\label{sec:pub}


This section considers the public allowance setting and shows the following:
\begin{theorem}\label{thm:public}
    In the public allowance setting, there exists a truthful and $(1+\epsilon)$-approximation deterministic mechanism for any parameter $\epsilon > 0$. Further, no deterministic truthful mechanism can guarantee $(1+\epsilon)$ approximation for arbitrarily tunable $\epsilon$.
\end{theorem}

To derive the algorithmic intuition and layout the techniques employed, we begin by stating the negative result.

\subsection{Lower Bound}\label{subsec:negative}

In this subsection, we first characterize truthful mechanisms for (public) allowance utility maximizers and then prove a lower bound. Use $ x_i(v_i, \vecv_{- i}) $ and $p_i(v_i, \vecv_{-i})$ to denote the obtained CTR and payment of bidder $i$ when she bids $v_{i}$ and other agents bid $\vecv_{- i}$.

\begin{lemma}\label{lem:public_truthful}
    A mechanism is truthful for (public) allowance utility maximizers only if the following two properties hold for each bidder $i\in [n]$ and any $\vecv_{- i}$: 
    \begin{enumerate}
        \item[--] (Allocation Monotonicity) As $v_i$ increases,  $x_i(v_i, \vecv_{- i})$ is non-decreasing. 
        \item[--] (Unit-Price Strict-Monotonicity) If bidder $i$ increases $v_i$ to $v'_i$ such that $x_i(v_i', \vecv_{- i}) > \frac{\gamma_i}{v_i}  > x_i(v_i, \vecv_{- i}) >0 $, the paid unit-price has to strictly increase: $\frac{p_i(v_i', \vecv_{-i})}{x_i(v_i', \vecv_{- i})} > v_i.$
    \end{enumerate}
\end{lemma}

\begin{proof}
    We first prove that the allocation has to be monotone. The basic idea is similar to the proof of Myerson's lemma~\cite{DBLP:journals/mor/Myerson81}.
    Consider a bidder $i$. For brevity, we drop the argument $\vecv_{-i}$ in this proof.  According to~\cref{def:DSIC}, a truthful mechanism must guarantee that for any bid $b_i$, 
    \begin{equation}\label{eq:unit-price}
         p_i(b_i) \leq b_i \cdot x_i(b_i);
    \end{equation}
    otherwise, bidder $i$ obtains a negative utility.
    Assume for contradiction there exist two values $\hv_i<v_i$ with $x_i(\hv_i) > x_i(v_i)$. Suppose that bidder $i$ has a private value $v_i$ but misreports $\hv_i$. Due to~\cref{eq:unit-price}, the payment constraint holds:
    \begin{equation*}
        \begin{aligned}
            p_i ( \hv_i) \leq \hv_i \cdot x_i(\hv_i) < v_i \cdot x_i(\hv_i),
        \end{aligned}
    \end{equation*}
    and therefore, the utility after misreporting is $v_i \cdot x_i(\hv_i) - (p_i(\hv_i)-\gamma_i )^+ $. According to the definition of truthfulness,
    \begin{equation}\label{eq:truth1}
        \begin{aligned}
            (p_i(\hv_i)-&\gamma_i )^+ - (p_i(v_i)-\gamma_i )^+   \geq v_i\cdot ( x_i(\hv_i) - x_i(v_i) ) > 0,
        \end{aligned}
    \end{equation}
    implying that $p_i(\hv_i) > \gamma_i$. Thus, 
    \begin{equation}\label{eq:payment}
    \begin{aligned}
         p_i(\hv_i) - p_i(v_i) & \geq (p_i(\hv_i)-\gamma_i )^+ - (p_i(v_i)-\gamma_i )^+ \\
         & \geq v_i\cdot ( x_i(\hv_i) - x_i(v_i) ).
    \end{aligned}
    \end{equation}

    Now suppose that bidder $i$ has a private value $\hv_i$ but misreports $v_i$. Due to~\cref{eq:unit-price} and~\cref{eq:payment}, the allocation still satisfies the payment constraint:
    \begin{equation*}
        \begin{aligned}
            p_i ( v_i) &\leq p_i(\hv_i) -  v_i\cdot ( x_i(\hv_i) - x_i(v_i) )\\
            & \leq p_i(\hv_i) -  \hv_i\cdot ( x_i(\hv_i) - x_i(v_i) ) \\
            & \leq \hv_i \cdot x_i(\hv_i)-  \hv_i\cdot ( x_i(\hv_i) - x_i(v_i) ) \\
            & = \hv_i \cdot x_i(v_i),
        \end{aligned}
    \end{equation*}
    and therefore, the utility after misreporting is $\hv_i \cdot x_i(v_i) - (p_i(v_i)-\gamma_i )^+ $. Again, due to the truthfulness, 
    \begin{equation}\label{eq:truth2}
        \begin{aligned}
            (p_i(\hv_i)-&\gamma_i )^+ - (p_i(v_i)-\gamma_i )^+   \leq \hv_i\cdot ( x_i(\hv_i) - x_i(v_i) )
        \end{aligned}.
    \end{equation}

    Combining~\cref{eq:truth1} and~\cref{eq:truth2}:
    \begin{equation*}
        \begin{aligned}
            v_i\cdot ( x_i(\hv_i) - x_i(v_i) ) \leq \hv_i\cdot ( x_i(\hv_i) - x_i(v_i) ),
        \end{aligned}
    \end{equation*}
    which contradicts our assumption that $x_i(\hv_i) > x_i(v_i)$ and proves the necessity of allocation monotonicity.

    For the second property, we also assume for contradiction that $p_i(v_i') \leq v_i \cdot x_i(v_i')$. Such an assumption implies that when bidder $i$ has a private value $v_i$ but misreports a higher $v_i'$, the payment constraint still holds. Since the mechanism is truthful, misreporting cannot increase the utility, which implies
    \begin{equation*}
        \begin{aligned}
            (p_i(v'_i)-\gamma_i )^+ - (p_i(v_i)-\gamma_i )^+   & \geq v_i\cdot ( x_i(v'_i) - x_i(v_i) ) > 0 . 
        \end{aligned}
    \end{equation*}
    Similarly, we have $p_i(v'_i) > \gamma_i$ and 
    \[ p_i(v'_i) \geq \gamma_i + v_i\cdot ( x_i(v'_i) - x_i(v_i) ).  \]

    Recalling the if-condition that $\gamma_i > v_i \cdot x_i(v_i)$, the above inequality implies that
    \begin{equation*}
        \begin{aligned}
            p_i(v'_i) &> v_i \cdot x_i(v_i) + v_i\cdot ( x_i(v'_i) - x_i(v_i) ) \\
            &= v_i \cdot x_i(v_i'),
        \end{aligned}
    \end{equation*}
    contradicting our assumption, and thus, completing the proof.

\end{proof}    

The lemma establishes two properties of truthful mechanisms for allowance utility maximizers. We observe that classical mechanisms such as VCG auction and Generalized Second Price (GSP) auction satisfy allocation monotonicity but fail in unit-price strict-monotonicity due to their treatment of ties. 
We build upon~\cref{lem:public_truthful} to prove a lower bound.

\begin{lemma}\label{lem:negative}
    For any deterministic truthful mechanism, there always exists a value $\epsilon >0$ such that its approximation ratio is at least $(1+\epsilon)$.
\end{lemma}

\begin{proof}
    We assume for contradiction that there exists a deterministic truthful mechanism $\cM$ such that for any $\epsilon>0$, its approximation ratio is less than $(1+\epsilon)$.
    Consider a scenario where two bidders are competing for a single slot. The allowance of each bidder is set to be $1/2$. Without loss of generality, we can assume that when both of them bid $1$, $\cM$ assigns the slot to the first bidder.

    Supposing that the second bidder increases her bid to $(1+\epsilon)$ for any $\epsilon$, $\cM$ has to assign the slot to her immediately; otherwise, its approximation ratio is at least $(1+\epsilon)$ when the two bidder's private value profile is $<1,1+\epsilon>$. Due to~\cref{lem:public_truthful}, in this case, the second bidder's payment should be strictly larger than $1$. Denote by $(1+\delta)$ the second bidder's payment. Clearly, $\delta\leq \epsilon$. 

    When the second bidder has a private value $(1+\epsilon)$, her current utility is $1/2+\epsilon-\delta$. Then the second bidder has the incentive to misreport a lower bid $(1+\epsilon')$ for $\epsilon' < \delta$. According to the assumption, $\cM$ still has to assign the slot to her but due to individual rationality, $\cM$ can only charge her at most $(1+\epsilon')$, implying that the bidder obtains a higher utility at least $(1/2+\epsilon-\epsilon')$. Thus, $\cM$ is untruthful, which contradicts the assumption and completes the proof.
    
\end{proof}

\subsection{Optimal Deterministic Mechanism}\label{subsec:public}

\begin{algorithm}[tb]
   \caption{\; Public Allowance Auction}
   \label{alg:public}
\begin{algorithmic}[1]
   \STATE {\bfseries Input:} bids $\cB=\{b_i\}_{i\in [n]}$, allowances $\vgamma = \{\gamma_i\}_{i\in [n]}$, CTRs $\valpha=\{\alpha_j\}_{j\in [k]}$ and parameter $\epsilon>0$.
   \STATE {\bfseries Output:} allocation $\Pi$ and payment $\cP$.
   \STATE Initialize $\pi_i\leftarrow k+1$ (the dummy slot) and $p_i\leftarrow 0 $, $\forall i\in [n]$.
   \FOR{each bidder $i\in [n]$}
   \STATE {\color{gray} /* Round $b_i$ down slightly such that it is an exponential multiple of $(1+\epsilon)$. */}
   \STATE Set $t_i \leftarrow \lfloor \log_{1+\epsilon} b_i \rfloor$ and $\barb_i=(1+\epsilon)^{t_i}$.
   \ENDFOR
   \STATE Reindex the bidders such that $\barb_1\geq \barb_2 \geq ... \geq \barb_n$ and break the ties in a fixed manner.
   \STATE {\color{gray} /* The Allocation Rule */}
    \STATE For each bidder $i \in [k] $, $\pi_i \leftarrow i$.
   \STATE {\color{gray} /* The Payment Rule */}
   \FOR{each bidder $i \in [k] $}
    \STATE Based on the allocation rule, compute a function $f_i(z)$ representing the CTR obtained by bidder $i$ when her rounded bid is $(1+\epsilon)^z$ and others' are $\mathbf{\barb}_{-i}$.
    \STATE Find the maximum value $z_{m}$ such that $(1+\epsilon)^{z_m} \cdot f_i(z_m) \leq \gamma_i$.
    \IF{ $t_i \leq z_{m}$ }
    \STATE $p_i \leftarrow \barb_i \cdot  f_i(t_i).  $ 
    \ELSE
    \STATE 
    $ p_i \leftarrow (1+\epsilon)^{z_m} \cdot f_i(z_m) + \sum\limits_{z=z_m+1}^{t_i} (1+\epsilon)^z \cdot \left(f_i(z)-f_i(z-1)\right). $
    \ENDIF
   \ENDFOR
   \STATE \textbf{return} $\Pi=\{\pi_i\}_{i\in [n]}$ and $\cP=\{p_i\}_{i\in [n]}$.
\end{algorithmic}
\end{algorithm}

\begin{figure}[htb]
  \centering
     \begin{subfigure}[b]{0.4\textwidth}
    \includegraphics[width=\textwidth]{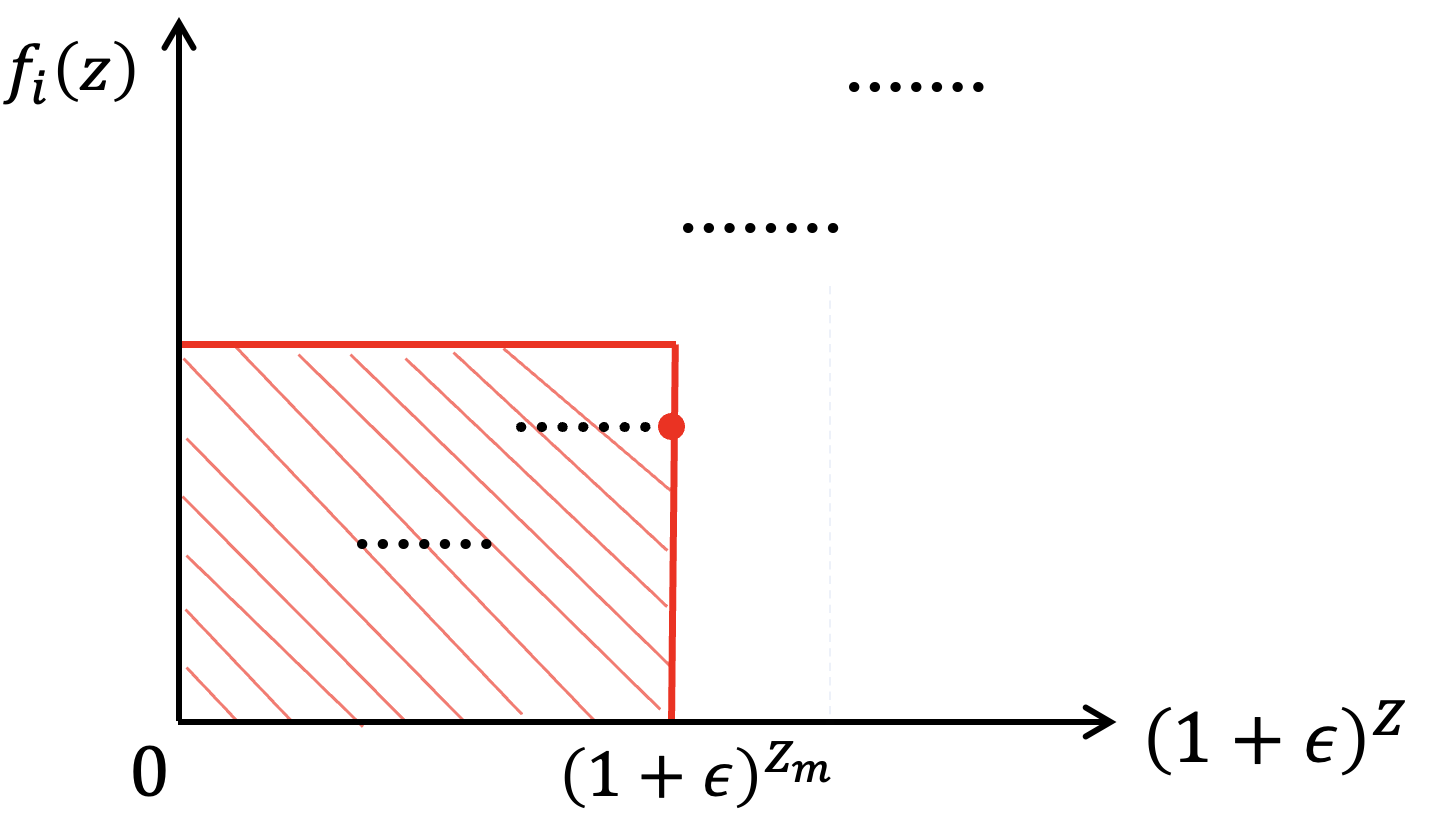}
    \caption{$\gamma_i > (1+\epsilon)^{z_m}\cdot f_i(z_m)$ }
  \end{subfigure}
  \begin{subfigure}[b]{0.4\textwidth}
    \includegraphics[width=\textwidth]{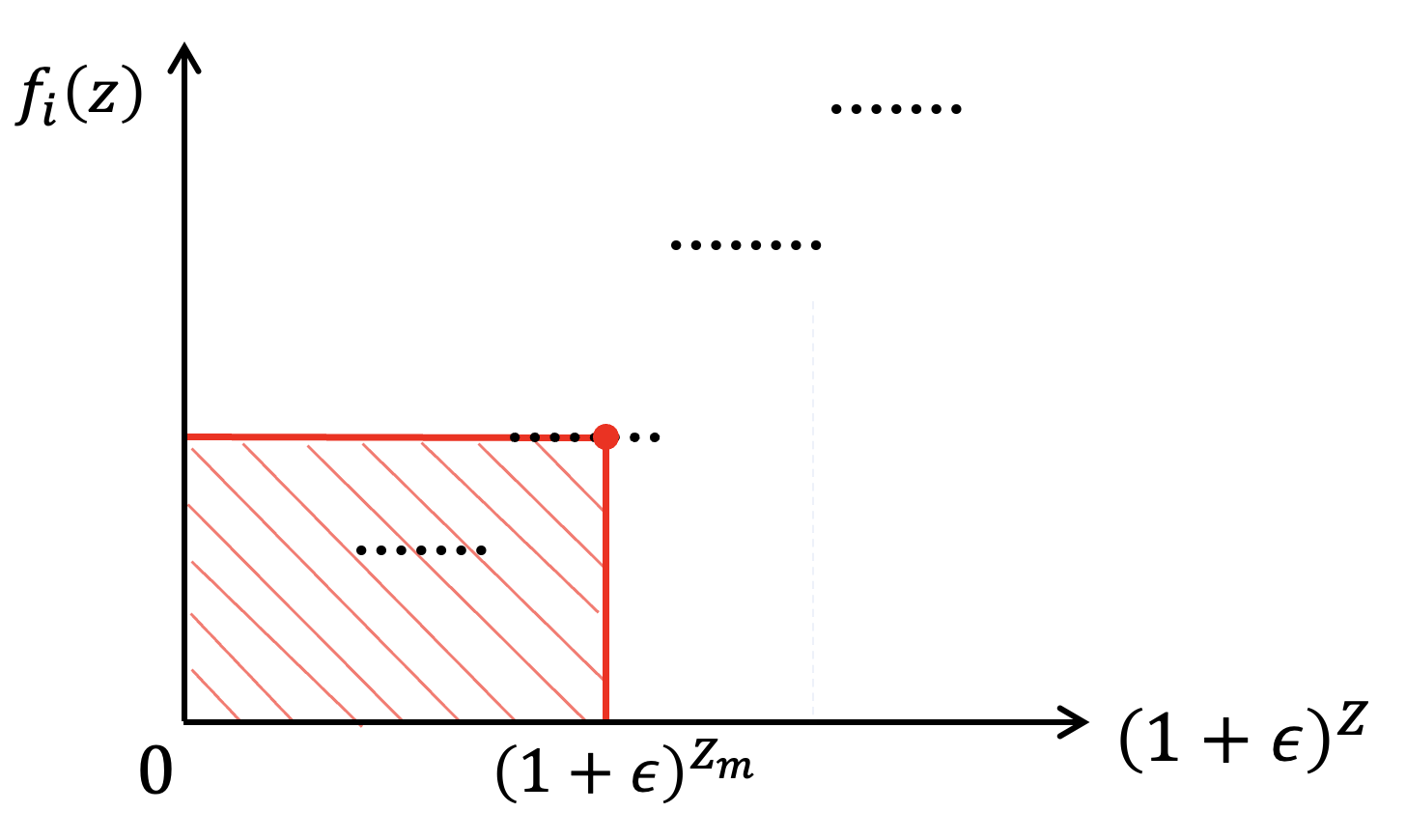}
    \caption{$\gamma_i = (1+\epsilon)^{z_m}\cdot f_i(z_m)$}
  \end{subfigure}
  \caption{ An illustration of function $f_i: Z \rightarrow \valpha$ and two potential cases of $z_m$. Note that the horizontal axis corresponds to $(1+\epsilon)^z$. The shaded area in each figure represents the allowance $\gamma_i$.}
  \label{fig:gamma}
\end{figure}
\begin{figure}[htb]
  \centering
  \begin{minipage}{0.9\textwidth}
     \begin{subfigure}[b]{0.5\textwidth}
    \includegraphics[width=\textwidth]{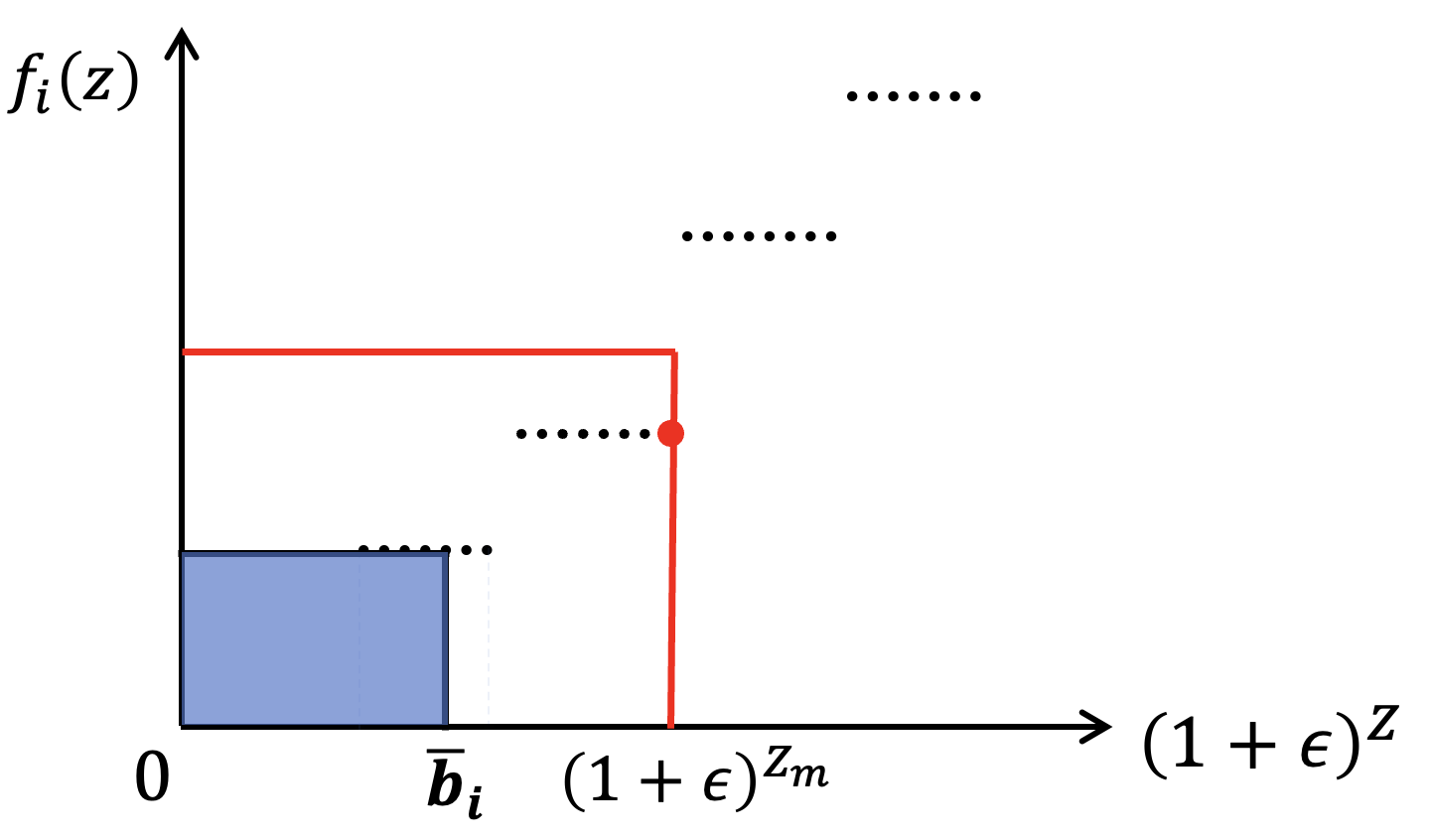}
    \caption{The payment when $\barb_i \leq (1+\epsilon)^{z_m}$}
  \end{subfigure}
  \quad
  \begin{subfigure}[b]{0.5\textwidth}
    \includegraphics[width=\textwidth]{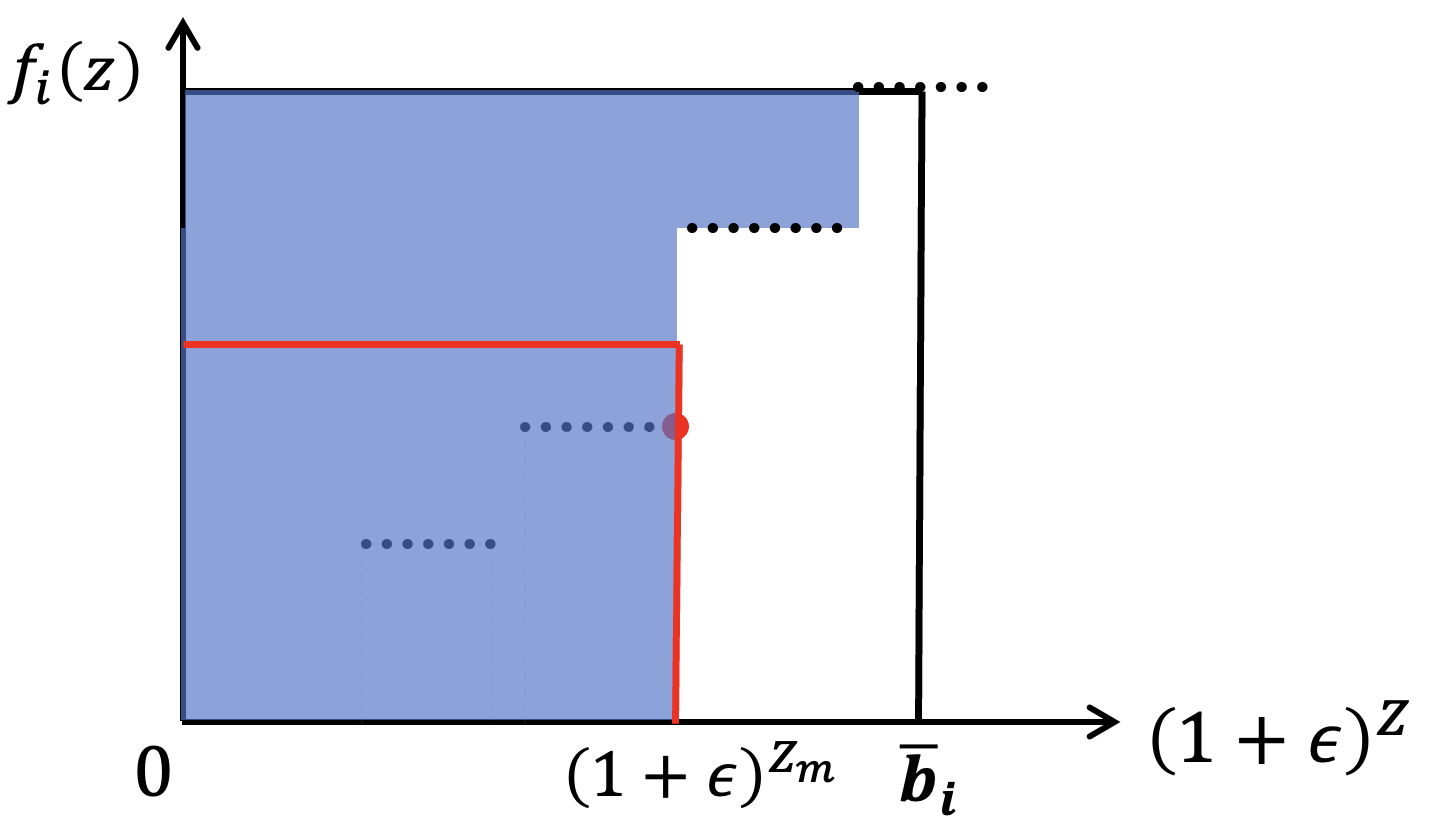}
    \caption{The payment when $\barb_i > (1+\epsilon)^{z_m}$}
  \end{subfigure}
  \quad
  \begin{subfigure}[b]{0.5\textwidth}
    \includegraphics[width=\textwidth]{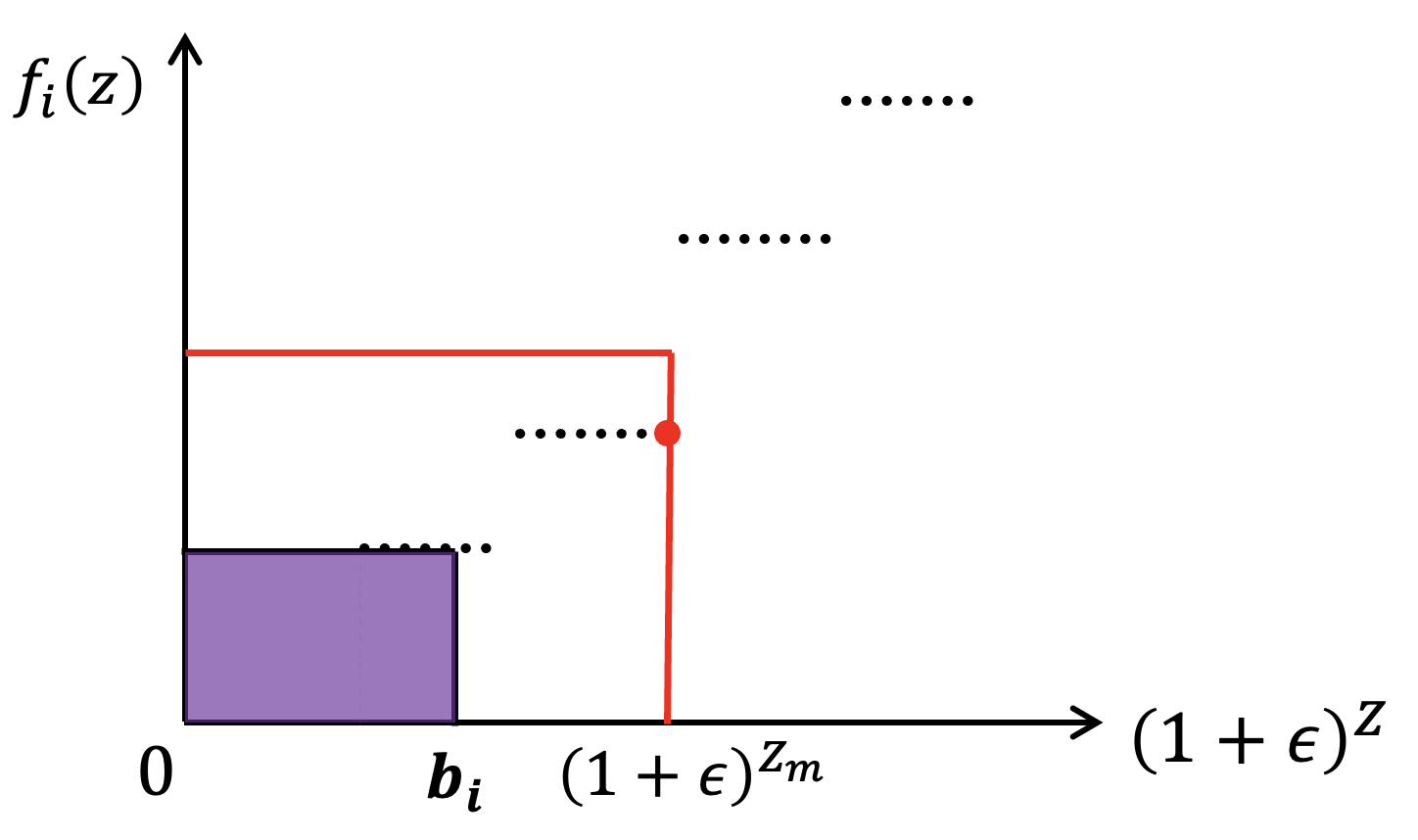}
    \caption{The utility when $\barb_i \leq (1+\epsilon)^{z_m}$}
    \label{subfig:u1}
  \end{subfigure}
  \quad
  \begin{subfigure}[b]{0.5\textwidth}
    \includegraphics[width=\textwidth]{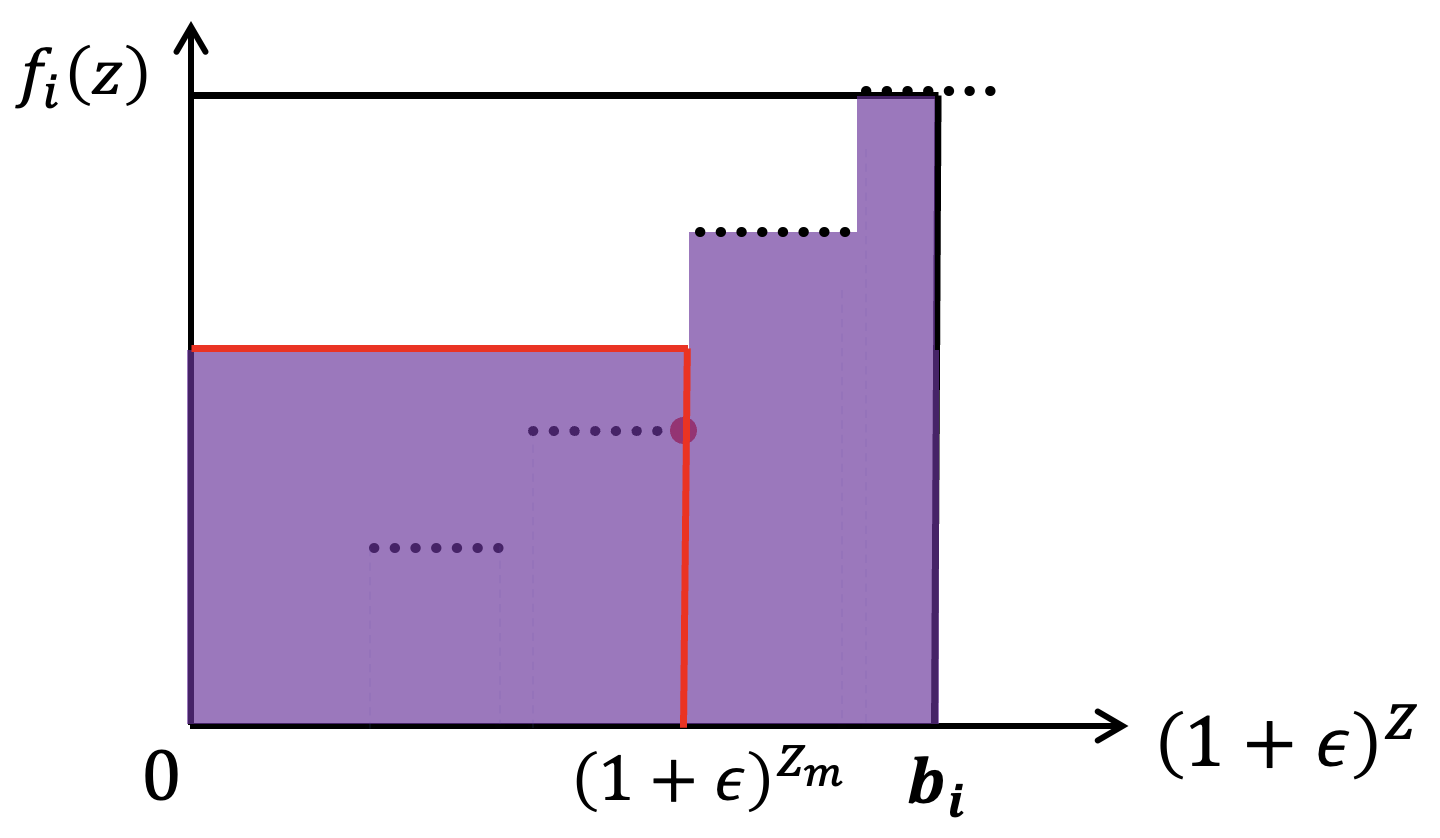}
    \caption{The utility when $\barb_i > (1+\epsilon)^{z_m}$}
    \label{subfig:u2}
  \end{subfigure} 
  \end{minipage}  
  \caption{An illustration of the payment function and the corresponding utilities in the case that $\gamma_i > (1+\epsilon)^{z_m}\cdot f_i(z_m) $. It is worth noting that in~\cref{subfig:u1} and~\cref{subfig:u2}, for an accurate illustration of the utility, we use $b_i$ instead of $\barb_i$. }
  \label{fig:pay_u}
\end{figure}

This subsection gives a truthful deterministic mechanism. As mentioned earlier, a significant factor leading to the lack of truthfulness in classical mechanisms under allowance utilities lies in the handling of tie-breaking situations.
To solve this issue, we first scale each bidder's bid to the nearest integer power of $(1+\epsilon)$, and then design the allocation and payment rules based on the scaled bids.

We state the details in~\cref{alg:public}. The mechanism allocates the slots to the bidders with top $k$ rounded bids in a fixed tie-breaking manner. According to this allocation rule, we compute function $f_i(z)$ and $z_m$ for each bidder $i$. Then we distinguish two cases based on $\barb_i$ and devise distinct payment functions for each case. Intuitively, the payment rule is running ``first pricing'' up to the allowance and then switching to the VCG-like payment. Two figures are provided for illustration in~\cref{fig:gamma} and~\cref{fig:pay_u} as follows.


\begin{proof}[Proof of~\cref{thm:public}]
    ~\cref{lem:negative} has already proven the negative side of the theorem. 
    We now show that~\cref{alg:public} is truthful and $(1+\epsilon)$ approximation. According to our allocation rule, for each slot, the assigned bidder's value is guaranteed to be at least $\frac{1}{1+\epsilon}$ of the value of the bidder who obtains the slot in the optimal solution,  establishing the approximation ratio. 
    
    The following analyzes the truthfulness. Consider a bidder $i$ with value $v_i$ and $t_i=\lfloor \log_{1+\epsilon} v_i \rfloor$. We distinguish two cases: (1) $t_i \leq z_m$, (2) $t_i > z_{m}$. 
    
    In the first case, submitting a truthful bid results in a utility of $ (1+\epsilon)^{t_i} \cdot f_i(t_i)$. If the bidder intends to alter the utility, she must misreport a bid $b_i'$ such that $t_i'$ is either $<t_i$ or $>t_i$. When $t_i'<t_i$, the utility strictly decreases due to the allocation monotonicity; when $t_i'>t_i$, regardless of whether $t_i'$ is larger than $z_m$ or not, the mechanism charges her a unit price of at least $ \min(z_m,t_i') $, violating the payment constraint and leading to a utility of $-\infty$. 

    In the second case, the truthtelling utility is 
    \begin{equation*}
        \begin{aligned}
            u_i =& \;(1+\epsilon)^{t_i} f_i(t_i) + \gamma_i - (1+\epsilon)^{z_m} f_i(z_m) -\sum\limits_{z=z_m+1}^{t_i} (1+\epsilon)^z \left(f_i(z)-f_i(z-1)\right) 
        \end{aligned}
    \end{equation*}

    If the bidder misreports a higher bid $b_i'$ such that $t_i'>t_i$, the new utility is
    \begin{equation*}
        \begin{aligned}
            u_i' =&\; (1+\epsilon)^{t_i} f_i(t_i') + \gamma_i - (1+\epsilon)^{z_m} f_i(z_m) -\sum\limits_{z=z_m+1}^{t_i'} (1+\epsilon)^z \left(f_i(z)-f_i(z-1)\right) \\
            =&\; u_i + (1+\epsilon)^{t_i} (f_i(t_i')-f_i(t_i)) - \sum\limits_{z=t_i+1}^{t_i'} (1+\epsilon)^z \left(f_i(z)-f_i(z-1)\right) \leq u_i.
        \end{aligned}
    \end{equation*}
    On the other hand, if the bidder misreports a lower bid $b_i'$ such that $t_i' < t_i$, we further distinguish two subcases: (\rom{1}) $t_i\geq z_m$; (\rom{2}) $t_i< z_m$. For subcase (\rom{1}), the new utility has
    \begin{equation*}
        \begin{aligned}
            u_i' =&\; (1+\epsilon)^{t_i} f_i(t_i') + \gamma_i - (1+\epsilon)^{z_m} f_i(z_m) -\sum\limits_{z=z_m+1}^{t_i'} (1+\epsilon)^z \left(f_i(z)-f_i(z-1)\right) \\
            =&\; u_i - (1+\epsilon)^{t_i} (f_i(t_i)-f_i(t_i')) + \sum\limits_{z=t_i'+1}^{t_i} (1+\epsilon)^z \left(f_i(z)-f_i(z-1)\right)  \leq  u_i.
        \end{aligned}
    \end{equation*}
    For subcase (\rom{2}), we have $u_i' = (1+\epsilon)^{t_i} f_i(t_i') $. Thus, 
    \begin{equation*}
        \begin{aligned}
            u_i - u_i'  =& \;(1+\epsilon)^{t_i} ( f_i(t_i)-f_i(t_i'))  + \gamma_i- (1+\epsilon)^{z_m} f_i(z_m) \\
            & -\sum\limits_{z=z_m+1}^{t_i} (1+\epsilon)^z \left(f_i(z)-f_i(z-1)\right) \geq 0.
        \end{aligned}
    \end{equation*}

    The analysis above collectively proves the truthfulness.

\end{proof}

\section{Private Allowance}\label{sec:private}

This section considers the setting where the allowance information $\vgamma$ is private. 
Compared to the public allowance setting, this scenario extends beyond the single-parameter environment, introducing increased complexity and posing greater challenges. 
In~\cref{alg:public}, the payment of bidder~$i$ is significantly dependent on the allowance $\gamma_i$. Such a strong correlation between the payment rule and $\vgamma$ indicates that when $\gamma_i$ becomes private, a bidder could benefit from misreporting a lower allowance.


To handle this issue, we explore \emph{$\vgamma$-independent} mechanisms, where both the allocation and payment rules are independent of the allowance profile of bidders, and show the following.

\begin{theorem}\label{thm:private_constant}
    In the private allowance setting, there exists a $\vgamma$-independent mechanism which is truthful and obtains a constant approximation ratio. 
\end{theorem}

To gain algorithmic intuition and streamline the analysis, we begin by considering the simplest case where $k=1$. Even in the presence of only one slot, developing a truthful and $\vgamma$-independent mechanism is not a trivial task.

\subsection{Single-Slot Auction}\label{subsec:private_single}

In this subsection, we consider the single-slot environment. Without loss of generality, assume that the slot's CTR is~$1$. Given our aim for a truthful mechanism, the auction must be truthful for any quasi-linear utility maximizer, i.e., the bidder with $\gamma_i=0$. For this type of bidders, Myerson’s Lemma~\cite{DBLP:journals/mor/Myerson81} dictates that the allocation rule must be monotone, and the payment rule is unique for a monotone allocation. Specifically, a bounded approximation mechanism should allocate the slot to one of the highest bidders and charge her the threshold bid.

Since we target a $\vgamma$-independent mechanism, we have to treat bidders as if they do not disclose their allowance profiles. Consequently, we are unable to distinguish whether a bidder employs quasi-linear utility or not. In light of this, to maintain truthfulness and a bounded approximation ratio, the classical second-price mechanism seems the most viable choice, where the highest bidder wins (with ties broken arbitrarily) and is charged the second-highest bid.

We examine the truthfulness of the second-price auction for bidders with $\gamma_i > 0$. We observe that these bidders are truthful in most cases, except in situations involving multiple highest bidders. When there is more than one highest bidder, quasi-linear utility maximizers have no incentive to lie because their utilities, whether they obtain the slot or not, consistently amount to~$0$. However, the scenario differs for other types of bidders. Upon securing the slot, they gain a non-zero utility, leading to a potential misreport of a higher bid to obtain the slot while paying an amount that adheres to the constraint. To tackle this challenge, we employ a scaling technique.

\begin{algorithm}[tb]
   \caption{\; $\vgamma$-Independent Single-Slot Auction}
   \label{alg:private_single}
\begin{algorithmic}[1]
   \STATE {\bfseries Input:} bids $\cB=\{b_i\}_{i\in [n]}$ and parameter $\epsilon>0$.
   \STATE {\bfseries Output:} allocation $\Pi$ and payment $\cP$.
   \STATE Initialize $\pi_i\leftarrow 2$ (the dummy slot) and $p_i\leftarrow 0 $, $\forall i\in [n]$.
   \STATE {\color{gray} /* Round $b_i$ down slightly such that it is an exponential multiple of $(1+\epsilon)$. */}
   \STATE For each bidder $i\in [n]$, set $\barb_i=(1+\epsilon)^{\lfloor \log_{1+\epsilon} b_i \rfloor}$.
   
   \STATE Use $\cS$ to denote the set of bidders with the highest $\barb$.
   \IF{ $|\cS| > 1$}
    \STATE Let $k$ be the smallest index in $\cS$. {\color{gray} /* Break ties in a fixed manner. */}
    \STATE Set $\pi_k \leftarrow 1$ and $p_k\leftarrow \barb_k$.
    \ELSE
    \STATE Let $\barb_k$ be the highest (rounded) bid and $\barb_\ell$ be the second-highest.
    \STATE Set $\pi_k \leftarrow 1$ and $p_k \leftarrow (1+\epsilon)\cdot \barb_\ell$.
   \ENDIF
   \STATE \textbf{return} $\Pi=\{\pi_i\}_{i\in [n]}$ and $\cP=\{p_i\}_{i\in [n]}$.
\end{algorithmic}
\end{algorithm}

\begin{theorem}\label{thm:private_single}
    For any $\epsilon >0 $,~\cref{alg:private_single} is truthful and obtains an approximation ratio of $(1+\epsilon)$.
\end{theorem}

\begin{proof}
    We begin by analyzing the truthfulness of the mechanism. 
    If bidder $i$ can win the auction with a truthful bid $b_i=v_i$, she obtains a non-negative utility: regardless of whether $|\cS|$ is $1$ or greater, the payment is always at most the private value. 
    Consequently, she has no incentive to lie because the allocation is monotone and the payment does not rely on $b_i$. 
    
    In the case that bidder $i$ loses the game with a truthful bid, we have $\barb_i\leq \barb_k$, where bidder $k$ is the winner.
    Clearly, the only chance to potentially increase the utility is by misreporting a higher bid $b_i'$ that enables her to win. If $\barb_i< \barb_k$, bidder $i$ would violate the payment constraint upon winning: 
    \[p_i \geq \barb_k \geq (1+\epsilon)\cdot \barb_i > b_i=v_i.\]
    If $\barb_i = \barb_k$, the bidder loses the game due to the fixed tie-breaking. Therefore, to obtain the slot, the misreported bid must satisfy $\barb_i' \geq (1+\epsilon)\cdot \barb_k$, which also breaches the constraint as now $|\cS'|=1$ and the payment becomes
    \[ p_i = (1+\epsilon)\cdot \barb_k = (1+\epsilon)\cdot \barb_i > b_i=v_i. \]

    The proof of the approximation ratio is straightforward. Given that the winner $k$ has the highest $\barb_k$, any other bidder's value is at most $(1+\epsilon)\cdot v_k$. As the single-slot auction has $\opt=\max_{i\in [n]}v_i$, the mechanism is $(1+\epsilon)$-approximation.
    
    
\end{proof}

\cref{thm:private_single} proves a constant approximation for single-slot auction. This result can further be extended to the multiple-slot setting. Specifically, if there exists a bidder making a significant contribution to the optimal social welfare $\opt$, we can choose to auction only the first slot and then apply~\cref{alg:private_single} to achieve a bounded approximation ratio:

\begin{corollary}\label{cor:private_single}
    For the multiple-slot setting, if there exists a bidder $i$ with $v_i \cdot \alpha_1 \geq \frac{\opt}{\rho} $ for some $\rho \geq 1$, ~\cref{alg:private_single} is truthful and obtains $\rho \cdot (1+\epsilon)$ approximation.
\end{corollary}


\subsection{Large Market Auction}

\cref{cor:private_single} gives a bounded approximation for the scenario where $\max_{i\in [n]}v_i \cdot \alpha_1 \geq \frac{\opt}{\rho} $. This subsection addresses the converse scenario – the large market case, where $\max_{i\in [n]}v_i \cdot \alpha_1 < \frac{\opt}{\rho}$. In a large market, no bidder or slot dominates.
Consequently, we can leverage the random sampling technique to develop a truthful mechanism with a bounded approximation ratio.

The mechanism is described in~\cref{alg:private_large}.
Initially, it partitions all bidders into two subsets $S^{(1)}$ and $S^{(2)}$, padding them with zeros until their sizes are at least $k$. Subsequently, it sets the prices of slots based on one subset and allows the bidders in the other subset to purchase them. The algorithmic intuition is derived from the following concentration lemma:

\begin{lemma}\label{lem:min_concen}
    Let $w_1\geq w_2 \geq ... \geq w_\ell$ be positive real numbers, such that the sum $w=w_1+...+w_\ell$ satisfies $w_1 < w/\rho$. We independently partition each number with equal probability into two subsets $ A $ and $B$. With probability at least $1/2$, there exists a matching $\cM$ between $A$ and $B$ such that  
    \[\sum_{(w_i,w_j)\in \cM} \min\{w_i,w_j\} > \frac{w}{3} \cdot \left(1-\frac{1}{\rho}\right).\]
\end{lemma}

\begin{algorithm}[tb]
   \caption{\; $\vgamma$-Independent Large Market Auction}
   \label{alg:private_large}
\begin{algorithmic}[1]
   \STATE {\bfseries Input:} bids $\cB=\{b_i\}_{i\in [n]}$, CTRs $\valpha=\{\alpha_j\}_{j\in [k]} $.
   \STATE {\bfseries Output:} allocation $\Pi$ and payment $\cP$.
   \STATE Initialize $\pi_i\leftarrow k+1$ (the dummy slot) and $p_i\leftarrow 0 $, $\forall i\in [n]$.
   \STATE Independently divide all the bidders with equal probability into set $S^{(1)}$ and $S^{(2)}$; subsequently, also with equal probability,  designate one of $\{ S^{(1)}, S^{(2)} \}$ as the pricing benchmark set $L$ and the other as the target bidder set $R$.
   \STATE For each slot $j\in [k]$, let $z_j$ be the $j$-th highest bid in $L$ (if $|L|<j$, let $z_j$ be $0$).
   \STATE Arrange the bidders in set $R$ in an arbitrarily fixed order. As each bidder $i$ arrives, offer all remaining slots to her at a unit price profile of $\left\{\frac{1}{2}z_j\right\}_{j\in [k]}$ and allow her to select the most profitable slot $j$. Set $\pi_i \leftarrow j$ and $p_i \leftarrow \frac{1}{2}z_j \cdot \alpha_j$. 
   \STATE \textbf{return} $\Pi=\{\pi_i\}_{i\in [n]}$ and $\cP=\{p_i\}_{i\in [n]}$.
\end{algorithmic}
\end{algorithm}

\begin{proof}
    
Define a partition $(A,B)$ where there exists a matching $\cM$ with \[\sum_{(w_j,w_j)\in \cM} \min\{w_i,w_j\} \geq \frac{w}{3} \cdot \left(1-\frac{1}{\rho}\right)\] as a ``good event''; otherwise, the partition is termed a ``bad event''.
Use $\cG$ and $\cB$ to denote all good events and bad events, respectively. Since we independently divide each number with equal probability, every event in $\cG\cup \cB$ occurs with equal probability. 
The basic idea of this proof is to show that for each bad event, we can always construct a unique good event for it, thereby implying that $|\cB| \leq |\cG|$ and $\cG$ occurs with probability at least $1/2$. 

Consider a bad event in $\cB$ and the corresponding partition $(A, B)$. Let $k$ denote the size of the smaller set among the two subsets. We match the $t$-th largest number $a_t$ in $A$ to the $t$-th largest number $b_t$ in $B$ for each index $t\in [k]$, obtaining a matching $\cM$ with size $k$. Furthermore, we define $ \barA := \{ a_t\in A | a_t\leq b_t \}$, which represents the elements in $A$ that are the minimum in their respective matching pairs. Similarly, we define $ \barB := \{ b_t\in B | b_t < a_t \}$.
Due to the definition of a bad event, we have 
\[\sum_{(w_i,w_j)\in \cM} \min \{w_i,w_j\} = \sum_{w_i\in \barA} w_i + \sum_{w_j\in \barB} w_j \leq \frac{w}{3} \cdot \left(1-\frac{1}{\rho}\right).\]

Define a set $C$ as the numbers that are not in $\barA \cup \barB$, i.e., $C:= (A\setminus \barA) \cup (B \setminus \barB)$. Sort the elements in $C$ from largest to smallest, and let $C_1$ and $C_2$ respectively denote the sets formed by the numbers positioned at odd and even indices. 
A new event $(A', B')$ is constructed by letting $A'=C_1 \cup \barA$ and $B'=C_2\cup \barB$. Clearly, for each bad event, such a new event $(A', B')$ is unique. 

We now show that the event $(A', B')$ is good. 
It is easy to observe that $|C_1| \geq |C_2|$. For each index $1\leq t\leq |C_1|$, we match the $t$-th largest number in $C_1$ to the $t$-th largest number in $C_2$. According to the construction of $C_1,C_2$, for each pair of this new matching $\cM'$, the smaller elements are always in $C_2$, and the sum of numbers in $C_2$ is at least the sum of numbers in $C_1$ minus the largest number $w_1$,  i.e., 
\begin{align*}
    \sum_{(w_i,w_j)\in \cM'} \min \{w_i,w_j\} = \sum_{w_j\in C_2} w_j \geq \sum_{w_i\in C_1} w_i - w_1 . 
\end{align*}

Due to the fact that \[\sum_{w_i\in C_1} w_j + \sum_{w_j\in C_2} w_j = w - (\sum_{w_i\in \barA} w_i + \sum_{w_j\in \barB} w_j) \geq  \left(\frac{2}{3}+\frac{1}{3\rho}\right)\cdot w \] and the large market assumption $w_1 < w/\rho$, we have
\[ \sum_{(w_i,w_j)\in \cM'} \min \{w_i,w_j\} > \frac{w}{2} \cdot \left(\frac{2}{3}+\frac{1}{3\rho} - \frac{1}{\rho}\right) = \frac{w}{3} \cdot \left(1-\frac{1}{\rho}\right), \]
demonstrating that $(A', B')$ is a good event and completing the proof.

\end{proof}

Intuitively, \cref{lem:min_concen} asserts that with high probability, there exists a matching between $k$ slots and $R$ such that if for each slot $j\in [k]$, we enforce the matched bidder in $R$ purchases it at a unit price of $z_j$ once her bid is no less than $z_j$, the total obtained payment is at least $\frac{1}{3} \cdot \left(1-\frac{1}{\rho}\right)$ of the optimal social welfare. 
However, the mechanism cannot enforce such actions as it would lead to untruthfulness. Therefore, we employ a technique of discounted sales to ensure both truthfulness and a constant approximation ratio simultaneously. We show the following theorem.

\begin{theorem}\label{thm:large_market}
    Under the large market assumption that $ \max_{i\in [n]}v_i \cdot \alpha_1 < \frac{\opt}{\rho} $, ~\cref{alg:private_large} is truthful and obtains an approximation ratio at most $48/(1-1/\rho)$. 
\end{theorem}

\begin{proof}

The proof of truthfulness is straightforward. For bidders in $L$, they are assigned anything, and therefore, misreporting their information cannot improve the utilities; while for bidders in $R$, they are also truthtelling because their reported information determines neither the arrival order nor the market prices.

Now we analyze the approximation ratio. Let $\Pi^*=\{\pi^*_i\}_{i\in [n]}$ be an optimal assignment. Viewing $b_i \alpha_{\pi^*_i}$ as $w_i$ in~\cref{lem:min_concen}, we have with probability at least $1/2$, there exists a matching $\cM$ between $S^{(1)}$ and $S^{(2)}$ such that \[\sum_{(i,h)\in \cM} \min\{ b_i \alpha_{\pi^*_i},b_h \alpha_{\pi^*_h}  \} > \frac{\opt}{3} \cdot \left(1-\frac{1}{\rho}\right)~. \]
Refer to such a partition $\{S^{(1)},S^{(2)}\}$ as a good partition. The following proof is conditioned on the occurrence of a good partition and analyzes the conditional expected performance of our algorithm. For brevity, we omit the notation for conditional expectations and use $\E[\cdot]$ to denote the conditional expectation.


Consider an assignment $\cA$ of the slots to bidders in $R$ as follows. For each slot $j$, we assign it to the bidder $i$ who is matched with the $j$-th highest bidder among $L$ in $\cM$. Since $L$ and $R$ are designated randomly, if enforcing that each bidder in $R$ either purchases the assigned slot at a unit price of $z_j$ or opts out, the total obtained payment $\cP(\cA)$ is at least $\sum_{(i,h)\in \cM} \min\{ b_i,b_h \}\cdot \alpha_j/2 $ in expectation. Due to the optimality of $\Pi^*$, we have that $b_i \alpha_{\pi^*_i}\geq b_h \alpha_{\pi^*_h}$ iff $b_i \geq b_h$, and for the $j$-th highest bidder $h$ in $L$, $\alpha_j \geq \alpha_{\pi^*_h}$. Thus, 
\begin{equation}\label{eq:large_eq1}
    \E[\cP(\cA)] \geq \frac{1}{2}\cdot \sum_{(i,h)\in \cM} \min\{ b_i \alpha_{\pi^*_i},b_h \alpha_{\pi^*_h}  \} > \frac{\opt}{6} \cdot \left(1-\frac{1}{\rho}\right).
\end{equation}

Let $\alg$ be the social welfare obtained by~\cref{alg:private_large}.
The last piece is to prove that $\E[\alg]$ is at least a constant factor of $\E[\cP(\cA)]$. 
Fix an arbitrary (good) pair of $(L,R)$. Let $\cJ$ be the set of slots $j$ that the assigned bidder $a(j)\in R$ has a bid at least $z_j$. We abuse the notion slightly and also let $a(j)$ represent the corresponding bid. 
Consider a slot $j\in \cJ$. We distinguish three cases according to its state in our mechanism: 
\begin{itemize}
    \item[(1)] it is purchased by some bidder.
    \item[(2)] no bidder purchases it and its corresponding bidder $a(j)$ picks a slot $j'$ with a higher CTR, i.e, $j'<j$. 
    \item[(3)] no bidder purchases it and its corresponding bidder $a(j)$ picks a slot $j'$ with a lower CTR, i.e., $j'>j$. 
\end{itemize}

Using $\cJ_i$ to denote the set of slots in Case (i), we have
\[\cP(\cA) = \sum_{j\in \cJ_1} z_j \cdot \alpha_j + \sum_{j\in \cJ_2} z_j \cdot \alpha_j + \sum_{j\in \cJ_3} z_j \cdot \alpha_j.  \]

Since all the slots in $\cJ_1$ are sold out,~\cref{alg:private_large} obtains a social welfare at least $\sum_{j\in \cJ_1} \frac{1}{2}z_j \cdot \alpha_j$. For the remaining two cases, we show that $z_j \cdot \alpha_j$ can be bounded by the bidder $a(j)$'s contribution in $\alg$, denoted as $a(j) \cdot \alpha_{j'}$. This claim is obvious for Case (2) since $a(j)\geq z_j$ and $\alpha_{j'}\geq \alpha_{j}$, while the proof for Case (3) requires leveraging the property of allowance utility. 

Consider a slot $j\in \cJ_3$. When bidder $a(j)$ arrives, this slot is available but the bidder picks another slot $j'$ with a lower CTR. According to the definition of allowance utility, we have
\begin{equation*}
    \begin{aligned}
        a(j)\cdot \alpha_j - \left(\frac{1}{2}z_j\cdot \alpha_j - \gamma_{a(j)}\right)^+ & \leq a(j)\cdot \alpha_{j'} - \left(\frac{1}{2}z_{j'}\cdot \alpha_{j'} - \gamma_{a(j)}\right)^+\\
        a(j)\cdot \alpha_j - \frac{1}{2}z_j\cdot \alpha_j &\leq a(j)\cdot \alpha_{j'} \\
        \frac{1}{2}z_j\cdot \alpha_j &\leq a(j)\cdot \alpha_{j'},
    \end{aligned}
\end{equation*}
where the second inequality is due to $\alpha_j \geq \alpha_{j'}$ and the last inequality uses the fact that $a(j) \geq z_j$. By summarizing the analysis of the three cases, we have
\begin{equation}\label{eq:large_eq2}
    \cP(\cA) \leq 4\cdot \alg.
\end{equation}

By combining \cref{eq:large_eq1}, \cref{eq:large_eq2}, and the probability of a good partition occurring, we can demonstrate that the expected approximation ratio of~\cref{alg:private_large} is at most $48/(1-1/\rho)$ under the large market assumption.

\end{proof}



\paragraph{Final Mechanism. } 
We state the final mechanism in~\cref{alg:private_final}, which was built on~\cref{alg:private_single} and~\cref{alg:private_large} . It is a simple random combination of the two mechanisms.

\begin{algorithm}[tb]
   \caption{\; $\vgamma$-Independent Auction}
   \label{alg:private_final}
\begin{algorithmic}[1]
   \STATE {\bfseries Input:} bids $\cB=\{b_i\}_{i\in [n]}$, CTRs $\valpha=\{\alpha_j\}_{j\in [k]} $ and parameter $\epsilon \in (0,1)$.
   \STATE {\bfseries Output:} allocation $\Pi$ and payment $\cP$.
   \BEGIN {With probability of $\sqrt{3}/(12+\sqrt{3})$}
   \STATE Run~\cref{alg:private_single} parameterized with $\epsilon$ to sell the first slot.
   \ENDBEGIN
   \BEGIN{With probability of $12/(12+\sqrt{3})$}
   \STATE Run~\cref{alg:private_large}.
   \ENDBEGIN
   \STATE \textbf{return} $\Pi=\{\pi_i\}_{i\in [n]}$ and $\cP=\{p_i\}_{i\in [n]}$.
\end{algorithmic}
\end{algorithm}

\begin{proof}[Proof of~\cref{thm:private_constant}]
    We show that~\cref{alg:private_final} is truthful and admits a constant approximation. The mechanism randomly picks one procedure and executes it. Due to the truthfulness guarantee of each procedure,~\cref{alg:private_final} is truthful. 

    Use $\alg_1, \alg_2$ to denote the social welfares obtained by the two procedures, respectively. According to the probability distribution, the final mechanism has
    \[ \E[\alg] = \frac{\sqrt{3}}{12+\sqrt{3}}\cdot \alg_1 + \frac{12}{12+\sqrt{3}}\cdot \alg_2. \]

    When the instance satisfies the large market assumption parameterized by $\rho=4\sqrt{3}+1$, due to~\cref{thm:large_market}, 
    \[ \E[\alg_2] \geq \frac{\sqrt{3}}{12+48\sqrt{3}} \cdot \opt; \]
    otherwise, the objective value of the first procedure is bounded by~\cref{cor:private_single}:
    \[ \alg_1 \geq \frac{1}{(4\sqrt{3}+1) (1+\epsilon)}  \cdot \opt. \]

    Combining the above two cases proves that ~\cref{alg:private_final} admits an approximation ratio of $(49+8\sqrt{3})(1+\epsilon)\approx 62.856$.

\end{proof}





\section{Uniform Price Auction in Large Markets}\label{sec:one_price}



This section considers uniform price auction design under the large market assumption that $\max_{i\in [n]}v_i \cdot \alpha_1 \geq \frac{\opt}{\rho} $ for a sufficiently large $\rho$. The basic idea is also leveraging the concentration property of large markets and employing the random sampling technique. Similar to~\cref{alg:private_large}, the mechanism also partitions all bidders into two subsets $S^{(1)}$ and $S^{(2)}$, each with an equal probability, and designate $L,R$ randomly. However, the uniform-price mechanism opts for a uniform price for all slots, departing from individual pricing strategies.

\begin{theorem}\label{thm:large_market_poly}
    Under the large market assumption that $ \max_{i\in [n]}v_i \cdot \alpha_1 < \frac{\opt}{\rho} $, there exists a $\gamma$-independent mechanism which is truthful in the private allowance setting and obtains an expected social welfare at least $\frac{3}{8} \left( 1-\sqrt{1-\frac{\rho}{3k}} \right)^2 \cdot \opt$,  where $k$ is the number of slots.
\end{theorem}

\begin{algorithm}[tb]
   \caption{\; Uniform Price Auction in Large Markets}
   \label{alg:private_large_uniform}
\begin{algorithmic}[1]
   \STATE {\bfseries Input:} bids $\cB=\{b_i\}_{i\in [n]}$, CTRs $\valpha=\{\alpha_j\}_{j\in [k]} $ and parameter $\beta \in (0,1)$.
   \STATE {\bfseries Output:} allocation $\Pi$ and payment $\cP$.
   \STATE Initialize $\pi_i\leftarrow k+1$ (the dummy slot) and $p_i\leftarrow 0 $, $\forall i\in [n]$.
   \STATE Pick the minimum index $t\in [k]$ such that $\sum_{j\in [t]}\alpha_j \geq \beta \cdot \sum_{j\in [k]}\alpha_j $.
   \STATE Randomly divide all the bidders with equal probability into set $S^{(1)}$ and $S^{(2)}$; subsequently, also with equal probability,  designate one of $\{ S^{(1)}, S^{(2)} \}$ as the pricing benchmark set $L$ and the other as the target bidder set $R$.
   \STATE Define the market price $Z$ to be the $t$-th highest bid in $\cL$; if $|\cL|<t$, $Z$ is $0$.
   \STATE Arrange the bidders in set $R$ in an arbitrarily fixed order. As each bidder $i$ arrives, offer all remaining slots to her at a unit price of $Z$ and allow her to select the most profitable slot $j$: $\pi_i \leftarrow j$ and $p_i \leftarrow Z \cdot \alpha_j$. 
   \STATE \textbf{return} $\Pi=\{\pi_i\}_{i\in [n]}$ and $\cP=\{p_i\}_{i\in [n]}$.
\end{algorithmic}
\end{algorithm}

The mechanism is described in~\cref{alg:private_large_uniform}. We start by introducing two pivotal lemmas \cref{lem:concen} and \cref{lem:subopt}. 

\begin{lemma}[~\cite{DBLP:journals/mor/ChenGL14}]\label{lem:concen}
    Let $a_1\geq a_2 \geq ... \geq a_\ell$ be positive real numbers, such that the sum $a=a_1+...+a_\ell$ satisfies $a_1 < a/36$. We select each number $a_1,...,a_\ell$ independently at random with probability $1/2$ each and let $b$ be the random variable equal to the sum of the selected numbers. Then
    \[ \Pr\left[\frac{a}{3} < b < \frac{2a}{3}\right] \geq \frac{3}{4}.  \]
\end{lemma}

This lemma captures the property of large markets where no bidder dominates and will be extensively utilized in the subsequent proofs.

Use  $\opt(\cS_1)$, $\opt(\cS_2)$, and $ \opt$ to respectively represent the social welfares obtained exclusively by auctioning bidders in $\cS_1$, exclusively by auctioning bidders in $\cS_2$, and by auctioning all bidders in $[n]$. 
We first show that both of them are at least a constant factor of $\opt$ with high probability.

\begin{lemma}\label{lem:subopt}
Under the large market assumption with $\rho \geq 36$, we have 
\[\min\{\opt(\cS_1), \opt(\cS_2)\} \geq \frac{1}{3}\opt\] with probability at least $3/4$. 
\end{lemma}

\begin{proof}
    Consider the allocation $\Pi^*$ in the optimal solution when auctioning all bidders in $[n]$. According to the definition of social welfare, we have
    \[ \sum_{i\in[n]} v_i\cdot \alpha_{\pi^*_i} = \opt. \]
    Due to the large market assumption, for any bidder $i\in [n]$
    \[ v_i\cdot \alpha_{\pi^*_i} < \frac{\opt}{36}. \]

    Both $\cS_1$ and $\cS_2$ can be viewed as generated by selecting each bidder independently at random with probability $1/2$ each. Thus, by~\cref{lem:concen},  
    \begin{equation}\label{eq:concen1}
        \Pr\left[ \frac{\opt}{3} < \substack{\sum\limits_{i\in\cS_1} v_i\cdot \alpha_{\pi^*_i} \\[1ex] \sum\limits_{i\in\cS_2} v_i\cdot \alpha_{\pi^*_i} } < \frac{2\opt}{3} \right] \geq \frac{3}{4}. 
    \end{equation}

    Observing that $\{\pi^*_i\}_{i\in \cS_1}$ is a feasible solution in the scenario that auctioning $\cS_1$, we have
    \begin{equation}\label{eq:concen2}
        \opt(\cS_1) \geq \sum\limits_{i\in\cS_1} v_i\cdot \alpha_{\pi^*_i}.
    \end{equation}
    Similarly,
    \begin{equation}\label{eq:concen3}
        \opt(\cS_2) \geq \sum\limits_{i\in\cS_2} v_i\cdot \alpha_{\pi^*_i} .
    \end{equation}

    Combining~\cref{eq:concen1},~\cref{eq:concen2} and~\cref{eq:concen3} completes the proof.
    
\end{proof}

\begin{proof}[Proof of~\cref{thm:large_market_poly}]

The proof of truthfulness is straightforward. For bidders in $\cL$, they will not be assigned anything, and therefore, misreporting their information cannot improve the utilities; while for bidders in $\cR$, they are also truthtelling because their reported information determines neither the arrival order nor the market price.

Now we analyze the approximation ratio. We first give a lower bound of the expected objective $\alg$ achieved by~\cref{alg:private_large} and then establish the relationship between the lower bound and $\opt$. 

Use $Z_1$ and $Z_2$ to denote the $t$-th highest bid among bidders in $\cS_1$ and $\cS_2$, respectively. Supposing that $Z_1\leq Z_2$, if the unit price of slots is set to be $Z_1$, at least $t$ bidders in $\cS_2$ can afford the cost. Hence, in our mechanism, with a probability of $1/2$, there are at least $t$ bidders in $\cR$ purchasing slots. According to the definition of $t$, we obtain a lower bound of our mechanism's objective value:
\begin{equation}\label{eq:private_lower_bound}
    \begin{aligned}
        \E\left[ \alg \right] & \geq \frac{1}{2}\cdot \sum_{j\in [t]} \alpha_j \cdot \E\left[\min\left\{ Z_1, Z_2 \right\}\right] \\
        & \geq \frac{\beta}{2} \cdot \valpha \cdot \E\left[\min\left\{ Z_1, Z_2 \right\}\right],
    \end{aligned} .
\end{equation}
where we abuse the notion slightly and let $\valpha := \sum_{j\in [k]} \alpha_j$.

Next, we relate the lower bound to $\opt$. Consider the price benchmark set $\cL$. Denote by $V_{max}:=\max_{i\in[n]} v_i$ the maximum private value. Since $Z$ is the $t$-th highest bid, we obtain an upper bound of $\opt(\cL)$ by replacing all bids greater than $Z$ with $V_{max}$ and all bids less than $Z$ with $Z$:
\begin{equation*}
    \begin{aligned}
          \left(\beta \cdot V_{max} + (1-\beta)\cdot Z\right) \cdot \valpha \geq \opt(\cL) \;, 
    \end{aligned}
\end{equation*}

\cref{lem:subopt} implies that with probability at least $3/4$, $\opt(\cL) \geq \frac{1}{3}\opt$. Due to the large market assumption that $\opt > \rho \cdot  V_{max} \cdot \alpha_1$, 
\begin{equation}\label{eq:private_large_opt}
    \begin{aligned}
       \left(\beta \cdot V_{max} + (1-\beta)\cdot Z\right) \cdot \valpha & \geq \frac{1}{3} \cdot \opt \\
       \left(\beta \cdot V_{max} + (1-\beta)\cdot Z\right) \cdot \valpha & \geq \frac{1}{3}\cdot \rho \cdot  V_{max} \cdot \alpha_1 \\
       \beta \cdot V_{max} + (1-\beta)\cdot Z & \geq \frac{1}{3}\cdot \rho \cdot  V_{max} \cdot \frac{\alpha_1}{\valpha} \\
       \beta \cdot V_{max} + (1-\beta)\cdot Z  & \geq \frac{1}{3}\cdot \rho \cdot  V_{max} \cdot \frac{1}{k} \\
       Z & \geq \frac{1}{1-\beta} \cdot \left(\frac{\rho}{3k} - \beta \right)\cdot V_{max}
    \end{aligned}.
\end{equation}

Note that~\cref{eq:private_large_opt} holds for both $Z_1$ and $Z_2$ with probability at least $3/4$. Combining it with~\cref{eq:private_lower_bound}, 
\begin{equation*}
    \begin{aligned}
        \E\left[ \alg \right] & \geq \frac{\beta}{2} \cdot \valpha \cdot \frac{3}{4} \cdot  \frac{1}{1-\beta} \cdot \left(\frac{\rho}{3k} - \beta \right)\cdot V_{max} \\
        & \geq \frac{3}{8} \cdot \frac{\beta}{1-\beta} \cdot \left(\frac{\rho}{3k} - \beta \right) \cdot \opt
    \end{aligned}. 
\end{equation*}

By setting $\beta = 1-\sqrt{1-\frac{\rho}{3k}}$, we get the best approximation:
\begin{equation*}\label{eq:private_large_final}
    \E [\alg] \geq \frac{3}{8} \left( 1-\sqrt{1-\frac{\rho}{3k}} \right)^2 \cdot \opt .
\end{equation*}

\end{proof}
\section{Conclusion}\label{sec:con}

The paper delves into sponsored search auctions in modern advertising markets, introducing a novel bidder utility model and exploring two distinct settings within the model. Specifically, for the public allowance setting, we demonstrate the achievability of an optimal mechanism, while in the private allowance setting, we propose a truthful mechanism with a constant approximation ratio and a uniform-price auction with bounded approximation in large markets.

There leave quite a lot of possibilities for future work. One direction is investigating the lower bounds of the allowance utility model. It is noteworthy that~\cite{DBLP:conf/aaai/LvZZLYLCW23} establishes a lower bound of $5/4$ in their mixed-bidder setting, which directly extends to our private allowance setting as a special case. Consequently, a constant gap exists between the upper and lower bounds in the private allowance setting, posing an open question regarding the potential closure of this gap.


\begin{credits}
\subsubsection{\ackname} This work is supported by the National Key Research Project of China under Grant No. 2023YFA1009402, the National Natural Science Foundation of China (No. 62302166), the Scientific and Technological Innovation 2030 Major Projects under Grant 2018AAA0100902, the Shanghai Science and Technology Commission under Grant No.20511100200, the Dean's Fund of Shanghai Key Laboratory of Trustworthy Computing, ECNU, and the Key Laboratory of Interdisciplinary Research of Computation and Economics (SUFE), Ministry of Education.

\end{credits}

\bibliographystyle{splncs04}
\bibliography{ref}

\begin{thebibliography}{10}
\providecommand{\url}[1]{\texttt{#1}}
\providecommand{\urlprefix}{URL }
\providecommand{\doi}[1]{https://doi.org/#1}

\bibitem{DBLP:conf/adkdd/AuerbachGS08}
Auerbach, J., Galenson, J., Sundararajan, M.: An empirical analysis of return
  on investment maximization in sponsored search auctions. In: AdKDD@KDD.
  pp.~1--9. {ACM} (2008)

\bibitem{baisa2017auction}
Baisa, B.: Auction design without quasilinear preferences. Theoretical
  Economics  \textbf{12}(1),  53--78 (2017)

\bibitem{DBLP:conf/nips/BalseiroDMMZ21}
Balseiro, S.R., Deng, Y., Mao, J., Mirrokni, V.S., Zuo, S.: Robust auction
  design in the auto-bidding world. In: NeurIPS. pp. 17777--17788 (2021)

\bibitem{DBLP:conf/sigecom/BalseiroDMMZ22}
Balseiro, S.R., Deng, Y., Mao, J., Mirrokni, V.S., Zuo, S.: Optimal mechanisms
  for value maximizers with budget constraints via target clipping. In: {EC}.
  p.~475. {ACM} (2022)

\bibitem{DBLP:conf/ijcai/BirmpasCCL22}
Birmpas, G., Celli, A., Colini{-}Baldeschi, R., Leonardi, S.: Fair equilibria
  in sponsored search auctions: The advertisers' perspective. In: {IJCAI}. pp.
  95--101. ijcai.org (2022)

\bibitem{DBLP:journals/mor/CaragiannisV21}
Caragiannis, I., Voudouris, A.A.: The efficiency of resource allocation
  mechanisms for budget-constrained users. Math. Oper. Res.  \textbf{46}(2),
  503--523 (2021)

\bibitem{DBLP:journals/mor/ChenGL14}
Chen, N., Gravin, N., Lu, P.: Truthful generalized assignments via stable
  matching. Math. Oper. Res.  \textbf{39}(3),  722--736 (2014)

\bibitem{DBLP:conf/icml/0002W0S0YD23}
Chen, Y., Wang, Q., Duan, Z., Sun, H., Chen, Z., Yan, X., Deng, X.: Coordinated
  dynamic bidding in repeated second-price auctions with budgets. In: {ICML}.
  Proceedings of Machine Learning Research, vol.~202, pp. 5052--5086. {PMLR}
  (2023)

\bibitem{DBLP:conf/www/DengLX20}
Deng, X., Lin, T., Xiao, T.: Private data manipulation in optimal sponsored
  search auction. In: {WWW}. pp. 2676--2682. {ACM} / {IW3C2} (2020)

\bibitem{DBLP:conf/nips/DengZ21}
Deng, Y., Zhang, H.: Prior-independent dynamic auctions for a value-maximizing
  buyer. In: NeurIPS. pp. 13847--13858 (2021)

\bibitem{DBLP:conf/sigecom/DevanurW17}
Devanur, N.R., Weinberg, S.M.: The optimal mechanism for selling to a budget
  constrained buyer: The general case. In: {EC}. pp. 39--40. {ACM} (2017)

\bibitem{DBLP:conf/wine/LiMM10}
Li, S., Mahdian, M., McAfee, R.P.: Value of learning in sponsored search
  auctions. In: {WINE}. Lecture Notes in Computer Science, vol.~6484, pp.
  294--305. Springer (2010)

\bibitem{DBLP:conf/wine/LuXZ23}
Lu, P., Xu, C., Zhang, R.: Auction design for value maximizers with budget and
  return-on-spend constraints. In: {WINE}. Lecture Notes in Computer Science,
  vol. 14413, pp. 474--491. Springer (2023)

\bibitem{DBLP:conf/aaai/LvZZLYLCW23}
Lv, H., Zhang, Z., Zheng, Z., Liu, J., Yu, C., Liu, L., Cui, L., Wu, F.:
  Utility maximizer or value maximizer: Mechanism design for mixed bidders in
  online advertising. In: {AAAI}. pp. 5789--5796. {AAAI} Press (2023)

\bibitem{DBLP:journals/mor/Myerson81}
Myerson, R.B.: Optimal auction design. Math. Oper. Res.  \textbf{6}(1),  58--73
  (1981)

\bibitem{DBLP:journals/eatcs/RoughgardenI17}
Roughgarden, T., Iwama, K.: Twenty lectures on algorithmic game theory. Bull.
  {EATCS}  \textbf{122} (2017)

\bibitem{szymanski2006impact}
Szymanski, B.K., Lee, J.S.: Impact of roi on bidding and revenue in sponsored
  search advertisement auctions. In: Second Workshop on Sponsored Search
  Auctions (2006)

\bibitem{DBLP:journals/corr/WilkensCN16}
Wilkens, C.A., Cavallo, R., Niazadeh, R.: Mechanism design for value
  maximizers. CoRR  \textbf{abs/1607.04362} (2016)

\bibitem{DBLP:conf/www/WilkensCN17}
Wilkens, C.A., Cavallo, R., Niazadeh, R.: {GSP:} the cinderella of mechanism
  design. In: {WWW}. pp. 25--32. {ACM} (2017)

\end{thebibliography}

\end{document}